\documentclass[10pt,conference,letterpaper]{ieeeconf}
\usepackage{epsfig}
\usepackage{amsmath}
\usepackage{amssymb,mathrsfs,dsfont}
\usepackage{enumerate}
\usepackage{color}
\usepackage{textcomp}
\usepackage{tikz}
\usepackage{pgfplots}
\usepackage{mathtools}
\usepackage{scrextend}
\usetikzlibrary{dsp,chains}

\usetikzlibrary{shapes,arrows}
\usetikzlibrary{positioning,shapes,shadows,arrows,scopes,decorations}
\usetikzlibrary{shapes.multipart}
\newtheorem{theorem}{Theorem}
\newtheorem{lemma}{Lemma}

\newtheorem{corollary}{Corollary}
\newtheorem{remark}{Remark}

\def\md{\mathbb}

\def\eps{\varepsilon}
\def\tn{\textnormal}
\def\wt{\widetilde}
\def\wh{\widehat}
\def\RealF{\md{R}}

\def\Expt{\md{E}}

\def\db{\mathrm{dB}}

\newcommand{\dfn}{ \stackrel{\tn{def}}{=} }

\newcommand{\moduloOp}[1] {\mathbb{M}_d[#1]}
\newcommand{\qfunc}[1] {Q\left(#1\right)}

\newcommand{\snr}{\mathrm{SNR}}
\newcommand{\bsnr}{\mathrm{S}\wt{\mathrm{N}}\mathrm{R}}
\newcommand{\dsnr}{\Delta\snr}
\newcommand{\Pe}{p_e}
\newif\ifShortVersion
\ShortVersiontrue

\IEEEoverridecommandlockouts

\renewcommand{\baselinestretch}{0.969}
\usepackage[noadjust]{cite}

\begin{document}

\title{\LARGE \bf The Gaussian Channel with Noisy Feedback: Near-Capacity Performance via Simple Interaction}

\author{Assaf~Ben-Yishai and Ofer~Shayevitz\thanks{The authors are with the Department of EE--Systems, Tel Aviv University, Tel Aviv, Israel \{assafbster@gmail.com, ofersha@eng.tau.ac.il\}. The work of O. Shayevitz was partially supported by the Marie Curie Career Integration Grant (CIG), grant agreement no. 631983.}}

\maketitle

\begin{abstract}
Consider a pair of terminals connected by two independent additive white Gaussian noise channels, and limited by individual power constraints. The first terminal would like to reliably send information to the second terminal, within a given error probability. We construct an explicit interactive scheme consisting of only (non-linear) scalar operations, by endowing the Schalkwijk-Kailath noiseless feedback scheme with modulo arithmetic. Our scheme achieves a communication rate close to the Shannon limit, in a small number of rounds. For example, for an error probability of $10^{-6}$, if the Signal to Noise Ratio ($\snr$) of the feedback channel exceeds the $\snr$ of the forward channel by $20\db$, our scheme operates $0.8\db$ from the Shannon limit with only $19$ rounds of interaction. In comparison, attaining the same performance using state of the art Forward Error Correction (FEC) codes requires two orders of magnitude increase in delay and complexity. On the other extreme, a minimal delay uncoded system with the same error probability is bounded away by $9\db$ from the Shannon limit. 
\end{abstract}

\section{Introduction}\label{sec:introduction}

Feedback cannot improve the capacity of point-to-point memoryless channels \cite{ShannonFeedback}. Nevertheless, noiseless feedback can significantly simplify the transmission schemes and improve the error probability performance, see e.g. \cite{S-K_partI,S-K_partII,horstein,PM_Transactions}. These elegant schemes fail however in the presence of arbitrarily small feedback noise, rendering them grossly impractical. This fact has been initially observed in \cite{S-K_partII} for the \textit{Additive White Gaussian Noise} (AWGN) channel, and further strengthened in \cite{KimNoisyAWGNFeedbackAllertor}. A handful of works have tackled the problem of noisy feedback as means for improving error performance, see e.g. \cite{ChanceLove,SatoYamamoto,BurnashevNoisyBSC,BurnashevNoisyAWGNISIT}. However, these works attain their superior error performance at the cost of a significant increase in complexity w.r.t. their noiseless feedback counterparts. There appears to be no \textit{simple} scheme (in the spirit of 
\cite{S-K_partII,horstein,PM_Transactions}) that is robust to feedback noise known hitherto.

Our work is therefore motivated by the following question: Does the simplicity of the infeasible noiseless feedback schemes extend itself to the more realistic noisy feedback setup, while still offering near-optimal performance? While the answer to this question appears to be  negative if one insists on approaching capacity in the usual sense (vanishing error probability in the limit of large delay), we answer it here in the affirmative under a fixed (but small) error probability criterion. Specifically, we consider the following setup: Two Terminals A and B are connected by pair of independent AWGN channels, and are limited by individual power constraints. The channel from Terminal A (resp. B) to Terminal B (resp. A) is referred to as the feedforward (resp. feedback) channel. Terminal A wishes to send bits to Terminal B, within a given bit error probability. The figure-of-merit we look at is the \textit{capacity gap}, which is the amount of excess $\snr$ required by our scheme over the minimal possible $\snr$ for an optimal Shannon scheme (of unbounded complexity), achieving the same bit rate and bit error probability. For this setup, we introduce and analyze a simple interactive scheme, that can operate near capacity. Our construction is based on the Schalkwijk-Kailath (S-K) noiseless feedback scheme \cite{S-K_partII} with active feedback, endowed with modulo arithmetic. Loosely speaking, our scheme is founded on the following observations: 
\begin{enumerate}[(1)]
\item The capacity gap (in $\db$) attained by the S-K scheme (for noiseless feedback) is inversely proportional to the number of iterations, and hence capacity is approached in a small number of rounds.  
\item The S-K scheme can be described as follows. Terminal A encodes and sends its message via \textit{Pulse Amplitude Modulation} (PAM), and in subsequent rounds, sends a scaled version of the estimation error of Terminal B (which is computable due to noiseless feedback), thereby exponentially decreasing the variance of the total estimation error. This scheme can operate using only \textit{passive} feedback. Alternatively, Terminal B could clearly employ \textit{active} feedback by transmitting its current estimate of the message, rather then its observations. This simple tweak is meaningless in the noiseless feedback case, yet turns out to be essential when feedback is noisy.
\item Suppose the S-K scheme is used when noise is present in the feedback channel. In each round, Terminal B knows the sum of the estimation error and the PAM message, whereas Terminal A knows the PAM message only. Describing the estimation error to Terminal A over the feedback channel is therefore a joint source-channel coding problem with side information at the receiver. Exploiting the side information could potentially yield a markedly better description of the estimation error. One simple way to reap this gain is by employing modulo arithmetic in the spirit of Tomlinson-Harashmia precoding \cite{Tomlinson,Harashima}. 
\item Following the above joint source-channel coding procedure, the estimation error of Terminal B becomes known at Terminal A, up to some excess additive noise induced by the noisy feedback. Due to the modulo-linearity of the operations, this excess noise can be effectively pushed into the forward channel.
\end{enumerate}

In a nutshell, our scheme operates as follows. Terminal A encodes and sends its message using PAM. In subsequent rounds, Terminal B computes its best linear estimate of the message, and feeds back a scaled version of that estimate, modulo a fixed interval. In turn, Terminal A employs a suitable modulo computation and obtains the estimation error, corrupted by excess additive noise. This quantity is then properly scaled and sent over the feedforward channel to Terminal B. After a fixed number of rounds, Terminal B decodes the message via a simple minimum distance rule. Loosely speaking, the scheme's error probability is dictated by the events of a modulo aliasing in one of the rounds, and the event where the remaining estimation noise in the last round exceeds half the minimum distance of the PAM constellation. The maximal number of rounds is limited by the need to control the modulo-aliasing errors. 

The resulting capacity gap (Theorem \ref{thrm:capgap}) consists of four terms: 1) An ``S-K term'' that is inversely proportional to the number of rounds; 2) A ``concatenated channel'' term, that corresponds to the decrease in $\snr$ incurred by trivially concatenating the forward and feedback channels, and is (roughly) inversely proportional to the excess $\snr$ of the feedback channel over the feedforward channel; 3) a ``modulo-aliasing'' term that stems from the error floor imposed by employing the modulo operation, and is (roughly) inversely proportional to the $\snr$ of the feedback channel; and 4) An auxiliary term that is (roughly) inversely proportional to the $\snr$ of the feedforward channel. 

As an example, for a bit error probability of $10^{-6}$, if the $\snr$ of the feedback channel exceeds the $\snr$ of the feedforward channel by $20\db$ (resp. $10\db$), our scheme operates at a capacity gap of $0.8\db$ (resp. $3.5\db$), with only $19$ (resp. $11$) rounds of interaction. This should be juxtaposed against two reference systems: On the one hand, state-of-the-art FEC codes attaining the same capacity gap and error probability, require roughly a two orders-of-magnitude increase in delay and complexity. On the other hand, the capacity gap attained by a minimal delay uncoded system with the same error probability, is at least $9\db$. 

The rest of the paper is organized as follows. The problem setup is introduced in Section \ref{sec:setup}. Necessary background including the capacity gap of uncoded PAM and an active feedback representation of the S-K scheme are given in Section \ref{sec:perlim}. Our new scheme is described in Section \ref{ourscheme}, and its performance is discussed in Section \ref{sec:main-res}. A detailed analysis of the scheme is provided in Section \ref{sec:proof}. Some numerical results and figures are given in Section \ref{results}. Implementation issues and the applicability of our scheme to real world scenarios are treated In Section \ref{sec:implem}. A discussion of the results and their context appears in Section \ref{sec:discussion}. 
 
\section{Setup}\label{sec:setup}
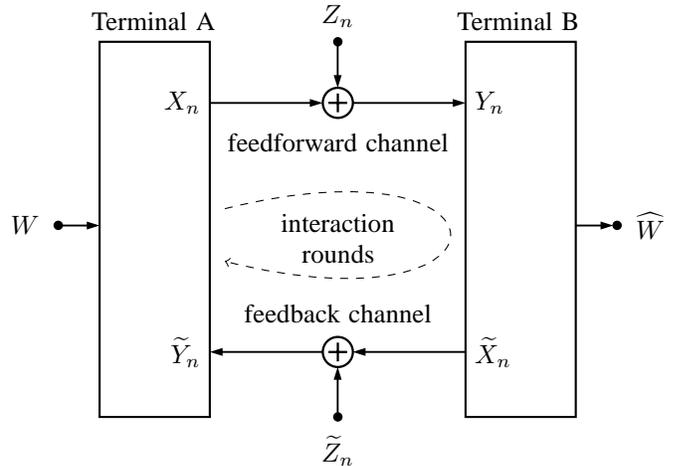
\begin{figure}
\centering
\begin{tikzpicture}
[every text node part/.style={align=center}]

	\matrix (m1) [row sep=2.5mm, column sep=5mm]
	{  

                \node[coordinate] (){};&
                \node[coordinate] (){};&
                \node[coordinate] (){};&
                \node[coordinate] (r1ur){};&
                \node[coordinate] (){};&
                \node[coordinate] (){};&
                \node[dspnodefull,dsp/label=above]   (z1){${Z}_n$};&
                \node[coordinate] (){};&
                \node[coordinate] (){};&
                \node[coordinate] (r2ur){};&
                \\
                \\
                \node[coordinate] (){};&
                \node[coordinate] (){};&
                \node[coordinate] (){};&
		\node[left]                  (x1) {$X_n$};   &
                \node[coordinate] (){};&
                \node[coordinate] (){};&
		\node[dspadder,dsp/label=below]  (a1) {feedforward channel}; &
                \node[coordinate] (){};&
                \node[coordinate] (){};&
		\node[right]                  (y1) {$Y_n$};  &

		\\ \\ \\ \\ \\
                \node[coordinate] (){};&
                \node[dspnodefull,dsp/label=left]   (win){$W$};&
                \node[right] (wout){};&
                \node[coordinate] (){};&
                \node[coordinate] (){};&
                \node[coordinate] (){};&
                \node[coordinate] (){};&
                \node[coordinate] (){};&
                \node[coordinate] (){};&
                \node[coordinate] (){};&
                \node[left] (whin){};&
                \node[dspnodefull,dsp/label=right]   (whout){$\widehat{W}$};&
                \\ \\ \\ \\ \\

                \node[coordinate] (){};&
                \node[coordinate] (){};&
                \node[coordinate] (){};&
                \node[left] (y2){$\widetilde{Y}_n$};&
                \node[coordinate] (){};&
                \node[coordinate] (){};&
                \node[dspadder,dsp/label=above]   (a2){feedback channel};&
                \node[coordinate] (){};&
                \node[coordinate] (){};&
                \node[right] (x2){$\widetilde{X}_n$};&
                \\  \\

                \node[coordinate] (){};&
                \node[coordinate] (){};&
                \node[coordinate] (r1ld){};&
                \node[coordinate] (){};&
                \node[coordinate] (){};&
                \node[coordinate] (){};&
                \node[dspnodefull,dsp/label=below]   (z2){$\widetilde{Z}_n$};&
                \node[coordinate] (){};&
                \node[coordinate] (){};&
                \node[coordinate] (){};&
                \node[coordinate] (r2ld){};&
                \node[coordinate] (){};&\\ \\
	};
        
          \draw[dspconn] (x1) -- (a1);
          \draw[dspconn] (a1) -- (y1);
          \draw[dspconn] (z1) -- (a1);
          \draw[dspconn] (x2) -- (a2);
          \draw[dspconn] (a2) -- (y2);
          \draw[dspconn] (z2) -- (a2);
          \draw[dspconn] (win) -- (wout);
          \draw[dspconn] (whin) -- (whout);

	\draw [color=black,thick](r1ur) rectangle (r1ld);
	\node at (-2.45,2.9) [] {\textrm{Terminal A}};

	\draw [color=black,thick](r2ur) rectangle (r2ld);
	\node at (2.4,2.9) [] {\textrm{Terminal B}};

       \def\xoffset{-10}
       \def\yoffset{1.2}
        \draw [->,dashed] (8.5+\xoffset,-0.8+\yoffset) .. controls (12.5+\xoffset,0.2+\yoffset) and (12.5+\xoffset,-2.5+\yoffset) .. (8.5+\xoffset,-1.5+\yoffset);
        \node at (0,0)  {{interaction} \\ {rounds}};

\end{tikzpicture}
\caption{\label{fig:blockdiagram}Block diagram of interactive coding over AWGN with noisy feedback}
\end{figure}

In the sequel, we use the following notations. For any number $x>0$, we write $x_{\db}\dfn10\log_{10}(x)$ to denote the value of $x$ in decibels. The Gaussian Q-function is  $\qfunc{x}\dfn(2\pi)^{-\frac{1}{2}}\int_{x}^{\infty}\exp\left( -u^2/2\right)du$, and $Q^{-1}(\cdot)$ is its functional inverse. We use the vector notation $x^n\dfn (x_1,\ldots,x_n)$. We write $f(x) = \mathrm{O}(g(x))$ for  $\mathrm{limsup}_{x\to\infty} \left|f(x)/g(x)\right| < \infty$.

Our problem setup is depicted in Fig.~\ref{fig:blockdiagram}. 
The feedforward and feedback channels connecting Terminal A to Terminal B and vice versa respectively, are AWGN channels given by 
\begin{align}
Y_n=X_n+Z_n ,  \\
\wt{Y}_n=\wt{X}_n+\wt{Z}_n . 
\end{align}
Where $X_n, Y_n$ (resp. $\wt{X}_n,\wt{Y}_n$) are the input and output of the feedforward (resp. feedback) channel at time $n$ respectively. The feedforward (resp. feedback) channel noise $Z_n\sim \mathcal{N}(0,\sigma^2)$ (resp. $\wt{Z}_n\sim \mathcal{N}(0,\wt{\sigma}^2)$) is independent of the input $X_n$ (resp. $\wt{X}_n$), and constitutes an i.i.d. sequence. The feedforward and feedback noise processes are mutually independent. 

Terminal A is in possession of a message $W\sim \textrm{Uniform}([M])$, to be described to Terminal B over $N$ rounds of communication. To that end, the terminals can employ an interactive scheme defined by a pair of functions $(\varphi,\wt{\varphi})$ as follows: At time $n$, Terminal A sends a function of its message $W$ and possibly of past feedback channel outputs over the feedforward channel, i.e., 
\begin{align}
  X_n=\varphi(W,\wt{Y}^{n-1}).
\end{align}
Similarly, Terminal B sends function of its past observations to Terminal A over the feedback channel, i.e., 
\begin{align}
  \wt{X}_n=\wt{\varphi}(Y^n).
\end{align}
\begin{remark}
The dependence of $\varphi$ and $\wt{\varphi}$ on $n$ is suppressed. In general, we allow these functions to further depend on common randomness shared by the terminals. A general interactive scheme can therefore be very complex; however, in what follows we will present and discuss a scheme that is very simple. We note in passing that our definition of the feedback transmission scheme is sometimes referred to as \textit{active feedback}; the term \textit{passive feedback} is usually reserved to the special case where $\wt{\varphi}(Y^n)=Y_n$. 
\end{remark}

We assume that Terminal A (resp. Terminal B) is subject to a power constraint $P$ (resp. $\wt{P}$), namely
\begin{align}
\sum_{n=1}^N\mathbb{E}(X_n^2) \leq N\cdot P, \quad \sum_{n=1}^N\mathbb{E}(\wt{X}_n^2) \leq N\cdot \wt{P} .
\end{align}
We denote the feedforward (resp. feedback) $\snr$ by $\snr\dfn\frac{P}{\sigma^2}$ 
(resp.  $\bsnr\dfn \frac{\wt{P}}{\wt{\sigma}^2}$). The ratio between the feedback $\snr$ and the feedforward $\snr$ is denoted by $\dsnr\dfn\frac{\bsnr}{\snr}$. Throughout this work, we assume that the feedback channel has excess $\snr$ over the feedforward
channel, i.e. $\dsnr>1$.

An interactive scheme $(\varphi,\wt{\varphi})$ is associated with a rate $R\dfn \frac{\log{M}}{N}$ and an error probability $\Pe$, which is the probability that Terminal B errs in decoding the message $W$ at time $N$, under the optimal decision rule.

The \textit{capacity gap} $\Gamma$ attained by the scheme is defined as follows. Recall that the Shannon capacity of the AWGN implies that the maximal rate achievable by any scheme (of unbounded complexity/delay, with or without feedback) under vanishing error probability, is given by 
\begin{equation}
\label{eq:AWNGcap}
C = \frac{1}{2}\log(1+\snr).
\end{equation}
Conversely, the minimal $\snr$ required to attain a rate $R$ is $2^{2R}-1$. The capacity gap is the excess $\snr$ required by the scheme, i.e., 
\begin{align}
\label{eq:capGapDef}
  \Gamma(\varphi,\wt{\varphi})  = \Gamma \dfn \frac{\snr}{2^{2R}-1}.
\end{align}
Note that if a nonzero bit/symbol error probability is allowed, then one can achieves rates exceeding the Shannon capacity \eqref{eq:AWNGcap}, and this effect should in principle be accounted for, to make the definition of the capacity gap fair. However, for small error probabilities the associated correction factor (given by the inverse of the corresponding rate distortion function) becomes negligible, and we therefore ignore it in the sequel.

\section{Preliminaries}\label{sec:perlim}
In this section, we describe the three building blocks underlying our interactive scheme. First, we discuss the performance of uncoded PAM transmission, and the associated capacity gap. We then describe the S-K scheme with active (noiseless) feedback, and derive the associated decay of the capacity gap as a function of the number of interaction rounds. Lastly, we briefly discuss the notations and properties of modulo arithmetic to be used in our scheme. 

\subsection{Uncoded PAM} 
\label{pambasic}
PAM is a simple and commonly used modulation scheme, where $2^R$ symbols are mapped (one-to-one) to the set $\{\pm 1\eta,\pm 3\eta,\cdots,\pm (2^R-1)\eta\}$. Canonically, the normalization factor $\eta$ is set so that the overall mean square of the constellation (assuming equiprobable symbols) is unity. A straightforward calculation yields $\eta = \sqrt{3/\left( 2^{2R}-1\right)}$. In the general case where the mean square of the constellation is constrained to be $P$, $\eta$ is replaced with $\eta\sqrt{P}$. 

It is easy to show that for an AWGN channel with zero mean noise of variance $\sigma^2$ and average input power constraint $P$, the probability of error incurred by the optimal detector is bounded by the probability that the noise exceed half the minimal distance of the PAM constellation, i.e.,
\begin{equation}
\label{eq:PeUncodedPAM}
\Pe<2\qfunc{\frac{\sqrt{P}\eta}{\sigma}}=2\qfunc{\sqrt{\frac{3\snr}{2^{2R}-1}}}.
 \end{equation}
Fixing the error probability $p_e$ and solving the inequality \eqref{eq:PeUncodedPAM} for $R$ yields:
\begin{equation}
\label{eq:capGap}
R>\frac{1}{2}\log\left(1+\frac{\snr}{\Gamma} \right),
\end{equation}
where 
\begin{align}
\label{eq:gammaPAM}
\Gamma_0(\Pe)\dfn\frac{1}{3}\left[Q^{-1}\left(\frac{\Pe}{2}\right)\right]^2.  
\end{align}
Comparing \eqref{eq:capGap} and \eqref{eq:AWNGcap}, we see that PAM signaling with error probability $\Pe$ admits a capacity gap of $\Gamma_0(\Pe)$. For a typical value of $\Pe=10^{-6}$, the capacity gap of uncoded PAM is $\Gamma_{0,\db} = 9\db$. 

Finally, we assume as usual that bits are mapped to PAM constellation symbols via  Gray labeling. The associated bit error probability can thus be bounded by
\begin{align}\label{eq:pb}
p_b<\frac{2}{R}\qfunc{\frac{\sqrt{P}\eta}{\sigma}}+2\qfunc{3\frac{\sqrt{P}\eta}{\sigma}} \approx \frac{\Pe}{R}.
\end{align}
where the approximation is becomes tight for small $\Pe$ due to the strong decay of the Q-function.

\subsection{The S-K Scheme with Active Feedback}
\label{s-kbasic}
Consider the setting of communication over the AWGN with noiseless feedback, i.e., where $\wt{\sigma}^2 = 0$. The S-K scheme with active feedback is described as follows. First, Terminal A  maps the message $W$ to a PAM constellation point $\Theta$. In the first round, it sends a scaled version of $\Theta$ satisfying the power constraint $P$. In subsequent rounds, Terminal B maintains an estimate $\wh{\Theta}_n$ of $\Theta$ given all the observation it has, and feeds it back to Terminal A. Terminal A then computes the estimation error $\eps_n\dfn \wh{\Theta}_n-\Theta$, and sends a properly scaled version of it to Terminal B. Formally:
\begin{enumerate}[(A)]
\item Initialization:
\begin{enumerate}[]
\item \textbf{Terminal A:} Map the message $W$ to a PAM point $\Theta$.  
\item \textbf{Terminal A $\Rightarrow$ Terminal B:} 
  \begin{itemize}
  \item Send $X_1=\sqrt{P}\Theta$
    \item Receive $Y_1=X_1+Z_1$
  \end{itemize}  

\item \textbf{Terminal B:} Initialize the $\Theta$ estimate\footnote{\label{fn1}Note that this is the minimum variance unbiased estimate of $\Theta$.}
to $\wh{\Theta}_1=\frac{Y_1}{\sqrt{P}}$. 
\end{enumerate}
\item Iteration:
\begin{enumerate}[]
\item \textbf{Terminal B $\Rightarrow$ Terminal A:} 
  \begin{itemize}
  \item Send the current $\Theta$ estimate: $\wt{X}_n= \wh{\Theta}_n$
  \item Receive $\wt{Y}_n=\wt{X}_n$
  \end{itemize}  
\item \textbf{Terminal A:} Compute the estimation error $\varepsilon_n = \wt{Y}_n - \Theta$.
\item \textbf{Terminal A $\Rightarrow$ Terminal B:} 
  \begin{itemize}
  \item Send the scaled estimation error $X_{n+1}=\alpha_n \varepsilon_n$, where $\alpha_n=\frac{\sqrt{P}}{\sigma_n}$ so that the input power constraint holds, and where $\sigma_n^2\dfn \Expt{\eps_n^2}$.
  \item Receive $Y_n=X_n+Z_n$
  \end{itemize}  
\item \textbf{Terminal B:} 
Update the $\Theta$ estimate\footref{fn1} $\wh{\Theta}_{n+1}=\wh{\Theta}_n-\wh{\varepsilon}_n$, where  
\begin{align}\label{eq:eps_est}
\wh{\varepsilon}_n=\beta_{n+1}Y_{n+1}  
\end{align}
is the \textit{Minimum Mean Square Error} (MMSE) estimate of $\varepsilon_n$, thus
\begin{equation}
\beta_{n+1}=\frac{\sqrt{P\sigma_n^2}}{P+\sigma^2}=\frac{\sigma_n}{\sigma}\cdot\frac{\sqrt{\snr}}{1+\snr}.
\end{equation}
\end{enumerate}
\item Decoding: 

At time $N$ the receiver decodes the message using a minimum distance decoder for $\wh{\Theta}_N$ w.r.t. the PAM constellation. 
\end{enumerate}

To calculate the error probability and rate attained by the S-K scheme, we note that $\eps_{n+1} = \eps_n - \wh{\eps}_n$. Computing the corresponding variance by plugging in the optimal values of $\alpha_n,\beta_n,Y_n$ yields
\begin{equation}
\label{eq:sigman}
\sigma_{n+1}^2=\frac{\sigma_{n}^2}{1+\snr}=\frac{1}{\snr\left(1+\snr\right)^{n}}.
\end{equation}
Since the power of $\Theta$ is normalized to unity, this is equivalent to signaling over an AWGN channel with $SNR_N=\sigma_N^{-2}$, i.e.
\begin{equation}
\label{eq:snreq}
\snr_N=\snr\cdot(1+\snr)^{N-1}.
\end{equation}
Plugging $\snr_N$ into \eqref{eq:PeUncodedPAM} and bounding the Q-function by 
$Q(x)<\frac{1}{2}\exp(-\frac{1}{2}x^2)$ gives:
\begin{equation}
\Pe<\frac{1}{2}\exp\left(-\frac{3}{2}\frac{\snr\cdot(1+\snr)^{N-1}}{2^{2NR}-1}\right).
\end{equation}
Plugging in the AWGN capacity \eqref{eq:AWNGcap} and removing the ``$-1$'' term, we obtain:
\begin{equation}
\Pe<\tfrac{1}{2}\exp\left(-\tfrac{3}{2}\tfrac{\snr}{1+\snr} \cdot 2^{2N(C-R)}\right).
\end{equation}
which is the well-known doubly exponential decay of the error probability of the S-K scheme. 

Let us now provide an alternative interpretation of the S-K scheme performance, in terms of the capacity gap attained after a finite number of rounds. Plugging $\snr_N$ in \eqref{eq:capGap} yields:
\begin{equation}
\label{eq:targetRate}
R>\frac{1}{2N}\log\left(1+\frac{\snr\cdot(1+\snr)^{N-1}}{\Gamma} \right).
\end{equation}
Plugging the resulting $R$ in the definition of the capacity gap \eqref{eq:capGapDef} 
and assuming $\snr\gg 1$ yields the following approximation for high $\snr$:
\begin{align}
  \Gamma_{\db}^{\text{S-K}}(\Pe,N) \approx \frac{\Gamma_{0,\db}(\Pe)}{N}.
\end{align}
This behavior is depicted by the dashed curve in Fig.~\ref{fig:resultsFigR4}.

\subsection{Modulo Arithmetic}
\label{modArithmetics}
We briefly overview basic notations and properties of modulo arithmetic. For a given  $d>0$, define the modulo function 
\begin{align}
\moduloOp{x}\dfn x-d\cdot\textrm{round}\left(\frac{x}{d}\right)  
\end{align}
where the \textrm{round$(\cdot)$} operator returns nearest integer to its argument\footnote{We arbitrarily define $\textrm{round}\left(k+\frac{1}{2}\right)=k+1$ for every integer $k$}. The following properties are easily verified: 
  \begin{enumerate}[(i)]
  \item $\moduloOp{x}\in[-\frac{d}{2},\frac{d}{2})$
  \item if $d_1+d_2\in[-\frac{d}{2},\frac{d}{2})$, then 
    \begin{align}\label{eq:modrel}
      \moduloOp{\moduloOp{x+d_1}+d_2-x}=d_1+d_2.
    \end{align}
    otherwise, a \textit{modulo-aliasing error} term of $kd\neq 0$ is added to the right-hand-side \eqref{eq:modrel}, for some integer $k$. 
  \item Let $V\sim\textrm{Uniform}([-\frac{d}{2},\frac{d}{2}))$. Then $\moduloOp{x+V}$ is uniformly distributed over $[-\frac{d}{2},\frac{d}{2})$ for any $x\in\RealF$. 
\item Therefore, $\mathbb{E}(\moduloOp{x+V})^2=\frac{d^2}{12}$. 
\end{enumerate}

\section{The Proposed Scheme}\label{ourscheme}
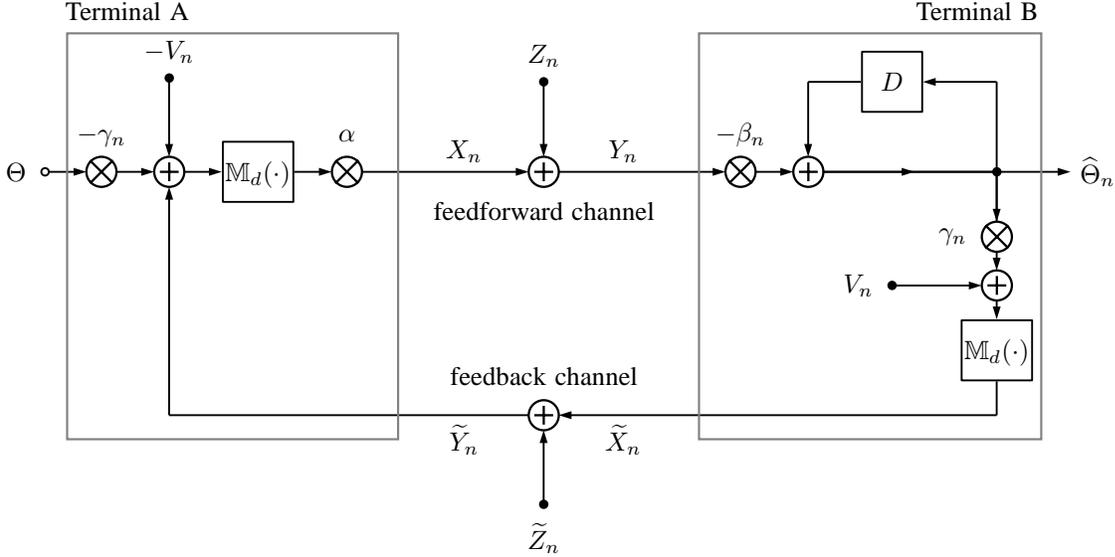
\begin{figure*}
\centering
\begin{tikzpicture}

	\matrix (m1) [row sep=2.5mm, column sep=5mm]
	{
                \node[coordinate] (){};&
                \node[coordinate] (){};&
                \node[dspnodefull,above] (mf1){$-V_{n}$};&
                \node[coordinate] (){};&
                \node[coordinate] (){};&
                \node[coordinate] (){};&
                \node[coordinate] (){};&
                \node[dspnodefull,dsp/label=above]   (mf4){${Z}_n$};&
                \node[coordinate] (){};&
                \node[coordinate] (){};&
                \node[coordinate] (){};&
                \node[coordinate] (mf6){};&
                \node[dspsquare] (mf7){$D$};&
		\node[coordinate] (mf8){}; &
                \node[coordinate] (){};&\\

		\node[dspnodeopen,dsp/label=left] (m0f) {$\Theta$};    &
		\node[dspmixer,dsp/label=above]   (m00) {$-\gamma_{n}$};          &
		\node[dspadder]                    (m01) {};          &
		\node[dspsquare, inner xsep=1pt]                   (m02) {$\mathbb{M}_d(\cdot)$}; &
		\node[dspmixer]                    (m03) {$\alpha$};    &
		\node[coordinate]                  () {};          &
		\node[above]                  () {$X_n$};          &
		\node[dspadder,dsp/label=below]  (m04) {{feedforward channel}}; &
		\node[above]                  () {$Y_n$};          &
		\node[coordinate]                  () {};          &
		\node[dspmixer]                    (m05) {$-\beta_n$};          &
		\node[dspadder]                    (m06) {}; &
		\node[coordinate]                  (m07) {};          &
		\node[dspnodefull] (m08) {};          &
		\node[right] (m09) {$\widehat{\Theta}_n$};          &
		\\
                \node[coordinate] (){};&
                \node[coordinate] (){};&
                \node[coordinate] (m11){};&
                \node[coordinate] (){};&
                \node[coordinate] (){};&
                \node[coordinate] (){};&
                \node[coordinate] (){};&
                \node[coordinate] (){};&
                \node[coordinate] (){};&
                \node[coordinate] (){};&
                \node[coordinate] (){};&
                \node[coordinate] (){};&
                \node[coordinate] (){};&
                \node[dspmixer,dsp/label=left]  (m18){$\gamma_n$};&
                \node[coordinate] (){};&\\

                \node[coordinate] (){};&
                \node[coordinate] (){};&
                \node[coordinate] (){};&
                \node[coordinate] (){};&
                \node[coordinate] (){};&
                \node[coordinate] (){};&
                \node[coordinate] (){};&
                \node[coordinate] (){};&
                \node[coordinate] (){};&
                \node[coordinate] (){};&
                \node[coordinate] (){};&
                \node[coordinate] (){};&
                \node[dspnodefull,dsp/label=left] (md7){$V_n$};&
                \node[dspadder]   (md8){};&
                \node[coordinate] (){};&
                \\

		--------------------------------------------------------------------
                \node[coordinate] (){};&
                \node[coordinate] (){};&
                \node[coordinate] (){};&
                \node[coordinate] (){};&
                \node[coordinate] (){};&
                \node[coordinate] (){};&
                \node[coordinate] (){};&
                \node[coordinate] (){};&
                \node[coordinate] (){};&
                \node[coordinate] (){};&
                \node[coordinate] (){};&
                \node[coordinate] (){};&
                \node[coordinate] (){};&
		\node[dspsquare, inner xsep=1pt]  (m28){$\mathbb{M}_d(\cdot)$}; &
                \node[coordinate] (){};&\\
                \node[coordinate] (){};&
                \node[coordinate] (){};&
                \node[coordinate] (m31){};&
                \node[coordinate] (){};&
                \node[coordinate] (){};&
                \node[coordinate] (){};&
                \node[below] (){$\widetilde{Y}_n$};&
                \node[dspadder,dsp/label=above]   (m34){{feedback channel}};&
                \node[below] (){$\widetilde{X}_n$};&
                \node[coordinate] (){};&
                \node[coordinate] (){};&
                \node[coordinate] (){};&
                \node[coordinate] (){};&
		\node[coordinate] (m38){}; &
                \node[coordinate] (){};&
                \\
                \\
                \node[coordinate] (){};&
                \node[coordinate] (){};&
                \node[coordinate] (){};&
                \node[coordinate] (){};&
                \node[coordinate] (){};&
                \node[coordinate] (){};&
                \node[coordinate] (){};&
                \node[dspnodefull,dsp/label=below]   (m44){$\widetilde{Z}_n$};&
                \node[coordinate] (){};&
                \node[coordinate] (){};&
                \node[coordinate] (){};&
                \node[coordinate] (){};&
                \node[coordinate] (){};&
		\node[coordinate] (){}; &
                \node[coordinate] (){};&\\ \\
	};
        
	\begin{scope}[start chain]
	\chainin (m00);
         \foreach \i in {1,2,3,4,5,6,9}
         {
		\chainin (m0\i) [join=by dspconn];
         }
        \chainin (m08);
        \chainin (m18) [join=by dspconn];
        \chainin (md8) [join=by dspconn];
        \chainin (m28) [join=by dspconn];
        \chainin (m38) [join=by dspline];
        \chainin (m34) [join=by dspconn];
        \chainin (m31) [join=by dspline];
        \chainin (m01) [join=by dspconn];
        \draw[dspconn] (m44) -- (m34);
        \draw[dspconn] (mf4) -- (m04);
        \draw[dspline] (m18) -- (mf8);
        \draw[dspconn] (mf8) -- (mf7);
        \draw[dspline] (mf7) -- (mf6);
        \draw[dspconn] (mf6) -- (m06);
        \draw[dspflow] (m06) -- (m08);
        \draw[dspconn] (mf1) -- (m01);
        \draw[dspconn] (m0f) -- (m00);
        \draw[dspconn] (md7) -- (md8);

	\end{scope}

	\draw [color=gray,thick](-7.1,-2) rectangle (-2.7,3.4);
	\node at (-6.3,3.7) [] {\textrm{Terminal A}};

	\draw [color=gray,thick](1.3,-2) rectangle (5.85,3.4);
	\node at (5,3.7) [] {\textrm{Terminal B}};

\end{tikzpicture}

\caption{\label{fig:blockdiagram}Block diagram of Our Scheme}
\end{figure*}

In what follows we assume that the terminals share a common random i.i.d sequence $\{V_n\}_{n=1}^N$ where $V_n\sim\textrm{Uniform}([-\frac{d}{2},\frac{d}{2}))$. Furthermore, we set $d=\sqrt{12\wt{P}}$ which guarantees that  $\mathbb{E}(\moduloOp{x+V_n})^2=\wt{P}$ for any $x\in\RealF$. Recall that the estimation of the PAM point at Terminal B and time instance $n$ is denoted by $\wh{\Theta}_n$, and the associated estimation error by $\eps_n\dfn\wh{\Theta}_n-\Theta$. 

Our scheme is described below. 
\begin{enumerate}[(A)]
\item Initialization:
\begin{enumerate}[]
\item \textbf{Terminal A:} Map the message $W$ to a PAM point $\Theta$.  
\item \textbf{Terminal A $\Rightarrow$ Terminal B:} 
  \begin{itemize}
  \item Send $X_1=\sqrt{P}\Theta$
    \item Receive $Y_1=X_1+Z_1$
  \end{itemize}  

\item \textbf{Terminal B:} Initialize the $\Theta$ estimate\footnote{\label{fn2}Note that this is the minimum variance unbiased estimate of $\Theta$.}
to $\wh{\Theta}_1=\frac{Y_1}{\sqrt{P}}$. 
\end{enumerate}
\item Iteration:
\begin{enumerate}[]
\item \textbf{Terminal B $\Rightarrow$ Terminal A:} 
  \begin{itemize}
  \item Given the $\Theta$ estimate $\wh{\Theta}_n$, compute and send 
    \begin{align}
      \wt{X}_n= \moduloOp{\gamma_n\wh{\Theta}_n+V_n}
    \end{align}
  \item Receive $\wt{Y}_n=\wt{X}_n + \wt{Z}_n$
  \end{itemize}  
\item \textbf{Terminal A:} Extract a noisy scaled version of
estimation error $\varepsilon_n$:
\begin{align}\label{eq:noisy_estimate}
\wt{\varepsilon}_n=\moduloOp{\wt{Y}_n-\gamma_n{\Theta}-V_n} 
\end{align}
Note that $\wt{\varepsilon}_n = \gamma_n\eps_n + \wt{Z}_n$, unless a \textit{modulo-aliasing error} occurs. 
 \item \textbf{Terminal A $\Rightarrow$ Terminal B:} 
  \begin{itemize}
  \item Send a scaled version of $\wt{\varepsilon}_n$:
$X_{n+1}=\alpha\wt{\varepsilon}_n$, where $\alpha$ is set so that to meet the input power constraint $P$ (computed later). 
  \item Receive $Y_n=X_n+Z_n$
  \end{itemize}  
\item \textbf{Terminal B:} 
Update the $\Theta$ estimate\footref{fn2} $\wh{\Theta}_{n+1}=\wh{\Theta}_n-\wh{\varepsilon}_n$, where  
\begin{align}
\wh{\varepsilon}_n=\beta_{n+1}Y_{n+1}  
\end{align}
is the MMSE estimate of $\varepsilon_n$. The optimal selection of $\beta_n$ is described in the sequel. 
\end{enumerate}
\item Decoding: 

At time $N$ the receiver decodes the message using a minimum distance decoder for $\wh{\Theta}_N$ w.r.t. the PAM constellation. 
\end{enumerate}

\section{Main Result}\label{sec:main-res}
Recall the capacity gap function $\Gamma_0(\cdot)$ of uncoded PAM given in \eqref{eq:gammaPAM}. Fix some target error probability $\Pe$.  Define: 
\begin{align}
\label{eq:penalty}
\nonumber\lambda &\dfn 3\left[Q^{-1}\left(\frac{\Pe}{4N}\right)\right]^{-2} \\
\nonumber\Psi_1 &\dfn 1+(\lambda\cdot\dsnr)^{-1}  \\ 
\nonumber\Psi_2 &\dfn \frac{1}{1-(\lambda\cdot\bsnr)^{-1}} \\
\Psi_3 &\dfn \frac{10/\ln{10}}{\snr\cdot \Psi_1^{-\frac{N-1}{N}}\Psi_2^{-\frac{N-1}{N}}\Gamma^{-\frac{1}{N}}_{0}\left(\frac{\Pe}{2}\right) -1} 
\end{align}
 
\begin{theorem}\label{thrm:capgap}
For a proper choice of parameters, the interactive communication scheme described in Section \ref{ourscheme} achieves in $N$ rounds an error probability $\Pe$ and a capacity gap $\Gamma_{\db}^*$ satisfying: 
\begin{align}
  \label{eq:capgapourshceme}
\Gamma_{\db}^*(\Pe,N) <  \tfrac{1}{N}\Gamma_{0,\db}(\tfrac{\Pe}{2}) + \tfrac{N-1}{N}\left(\Psi_{1,\db}  + \Psi_{2,\db}\right)  + \Psi_3 \phantom{\frac{1}{\frac{1}{2}}}
\end{align}
\end{theorem}
We prove this theorem is Section \ref{sec:proof}.

\begin{remark}
$\lambda$ is a factor that encapsulates the cost of controlling the modulo-aliasing error, as seen below. It decreases with a decreasing $\Pe$.  
\end{remark}
\begin{remark}\label{rem:concat}
$\Psi_1$ is a penalty term roughly corresponding to the decrease in $\snr$ incurred by trivially concatenating the forward and feedback channels. To see this, consider the  \textit{concatenated channel} from $X_n$ to $\wt{Y}_n$ where Terminal B performs simple linear scaling to meet the power constraint $\wt{P}$, i.e. $\wt{X}_n=\sqrt{\frac{\wt{P}}{P+\sigma_2}} Y_n$. The $\snr$ of this channel is 
\begin{align}
\snr_\textrm{concatenated}\dfn\frac{\snr\cdot\bsnr}{\snr+\bsnr+1},
\end{align}
hence, the associated $\snr$ loss w.r.t. the feedforward channel is 
\begin{align}
\frac{\snr}{\snr_\textrm{concatenated}} = 1+\frac{1}{\dsnr}+\frac{1}{\bsnr}\approx 1+\frac{1}{\dsnr}.
\end{align}
This latter expression is very similar to $\Psi_1$, with the exception of the additional $\lambda$ factor. Hence, loosely speaking, $\Psi_1$ encapsulates the inherent loss due to essentially employing a feedback scheme over the concatenated channel, together with a feedback power reduction by the amount of $\lambda$ used to avoid modulo-aliasing errors. This loss vanishes for a fixed $\Pe$ as $\dsnr$ increases. However, if $\dsnr$ is fixed, this term does not vanish in the limit of high $\snr$. 
\end{remark}

\begin{remark}\label{rem:Psi2}
$\Psi_2$ can be interpreted as a penalty term stemming from the modulo-aliasing error endemic to the system, due to the presence of feedback noise in the modulo operations at Terminal A. For a fixed $\bsnr$, the minimal value of $\lambda$ supported by our scheme is given by $\bsnr^{-1}$, which in turn dictates the minimal error probability that can be attained. Due to this error floor, our scheme cannot achieve any rate in the usual sense. The loss of $\snr$ incurred by $\Psi_2$ vanishes for any fixed error probability $\Pe$ as $\bsnr$ increases. 
\end{remark}

\begin{remark}
 $\Psi_3$ is an additional penalty term (already in logarithmic scale), that result from the fact that we consider the capacity gap in terms of $\snr$ ratios, whereas the explicit term arising from the capacity formula is related to $\log{(1+\snr)}$ rather than $\log{(\snr)}$. Note that $\Psi_3 = \mathrm{O}\left(\snr^{-1}\right)$. 
\end{remark}

\begin{corollary}[High $\snr$ behavior]
Let $\dsnr$ and $\Pe$ be fixed. The capacity gap attained by our scheme for $\snr$ large enough, can be approximated by
\begin{align}
\label{eq:capgapapprox}
\Gamma^*_{\db}(\Pe,N) \approx \tfrac{1}{N}\Gamma_{0,\db}(\tfrac{\Pe}{2}) + 
\tfrac{N-1}{N}\left[1+\frac{1}{\lambda\dsnr}\right]_\db.
\end{align}
The first term is roughly the capacity gap of the S-K scheme with noiseless feedback. The second term pertains to the $\snr$ loss w.r.t. a concatenated channel as well as modulo-aliasing errors, as discussed in Remark \ref{rem:concat}.
\end{corollary}

\begin{remark}
Note that there is a ``low $\snr$'' regime (related also to the target error probability or to $\dsnr$), where the loss terms $\Psi_{1,\db}+\Psi_{2,\db}$ are larger than say $\Gamma_{0,\db}(\Pe)$. In that case, setting $N=1$, namely using an uncoded system with no interaction, is the optimal choice of parameters for our scheme\footnote{The reason we get the looser $\Gamma_{0,\db}\left(\frac{\Pe}{2}\right)$ term in \eqref{eq:capgapapprox} is for brevity of exposition; a more accurate trade-off is given in the next section.}. As we shall see however, for many practical values of $\snr,\bsnr$ and $\Pe$, interaction results in significant gains.     
\end{remark}

\section{Proof of Main Result}\label{sec:proof}
In Subsection \ref{s-kbasic} we analyzed the error probability of S-K with noiseless feedback, relying on the fact that all the noises are jointly Gaussian, including the noise $\varepsilon_N$ experienced by the PAM decoder. To that end, we were able to directly use the error probability analysis of simple PAM over AWGN discussed in Subsection \ref{pambasic}. 

In the noisy feedback case however, the non-linearity induced by modulo operations at both terminals induce a non-Gaussian distribution of $\varepsilon_N$. An analysis of the decoding error based on the actual distribution of $\varepsilon_N$ is very 
involved. Yet, an upper bound can be derived via a simple coupling argument described below. 

Recall that Terminal A computes $\wt{\eps}_n$, a noisy scaled version of the estimation error of Terminal B, via a modulo operation \eqref{eq:noisy_estimate}. For any $n\in\{1,\ldots, N-1\}$ we define $E_n$ as the event where this computation results in a modulo-aliasing error, i.e.,
\begin{equation}
\label{eq:moderror}
E_n\dfn\{ \gamma_n\varepsilon_n+\wt{Z}_n\notin[-\tfrac{d}{2},\tfrac{d}{2})\}.
\end{equation}
Furthermore, we define $E_N$ as the PAM decoding error event:
\begin{equation}
E_N=\{ \varepsilon_N\notin[-\tfrac{d_{min}}{2},\tfrac{d_{min}}{2})\},
\end{equation}
where $d_{min}$ is the PAM constellation minimal distance. As mentioned above, the distribution of  $\varepsilon_N$ is not Gaussian due to the nonlinearity introduced by the modulo operations. To circumvent this, we consider the following upper bound for the error probability: 
\begin{align}\label{eq:err_union}
\Pe < \Pr\left(\bigcup_{n=1}^NE_n \right).
\end{align}
The inequality stems from the fact that a modulo-aliasing error does not necessarily cause a PAM decoding error.

To proceed, we define the \textit{coupled system} as a system that is fed by the same message and experiences the (sample-path) exact same noises, with the only difference being that no modulo operations are implemented at neither of the terminals. Clearly, the coupled system violates the power constraint at Terminal B. However, given the message $W$, all the random variables in the coupled system are jointly Gaussian, and in particular, the estimation errors $\varepsilon_n$ in that system are Gaussian for $n=1,\ldots,N$. Moreover, it is easy to see that the estimation errors are \textit{sample-path identical} between the original system and the coupled system until the first modulo-aliasing error occurs. To be precise: 
\begin{lemma}
Let  $\wt{\Pr}$ denote the probability operator in for the coupled process. Then for any $N>1$:
\begin{equation}
\Pr\left(\bigcup_{n=1}^N E_n \right) =  \wt{\Pr}\left(\bigcup_{n=1}^N E_n \right).
\end{equation}
 \end{lemma}

\begin{proof}
\begin{equation}
\Pr\left(\bigcup_{n=1}^N E_n \right) =  \Pr(E_1) +  \sum_{n=2}^{N}\Pr\left(\bigcap_{i=1}^{n-1}\left(E_i^C\right)\bigcap E_n\right).
\end{equation}
Moreover, for any $i\in\{2,\ldots,N\}$ 
\begin{equation}
\Pr\left(\bigcap_{i=1}^{n-1}\left(E_i^C\right)\bigcap E_n\right) = 
\wt{\Pr}\left(\bigcap_{i=1}^{n-1}\left(E_i^C\right)\bigcap E_n\right)
\end{equation}
and trivially $\Pr(E_1) = \wt{\Pr}(E_1)$. 
\end{proof}
Combining the above with \eqref{eq:err_union} and applying the union bound in the coupled system, we obtain
\begin{align}\label{eq:union_bound}
\Pe\leq\sum_{n=1}^{N}\wt{\Pr}\left( E_n\right).
\end{align}
Calculating the above probabilities now involves only scalar Gaussian densities, which significantly simplifies the analysis.

\subsection{Calculation of the Parameters}
We set $\gamma_n$ such that $\wt{\Pr}(E_1) =\cdots = \wt{\Pr}(E_{n-1}) \dfn p_m$, for some $p_m$ small enough to be set later. Specifically, recalling the definition \eqref{eq:moderror} and that $d=\sqrt{12\wt{P}}$, and since $\wt{\eps}_n = \gamma_n\varepsilon_n+\wt{Z}_n$ in the coupled system is Gaussian, we obtain the following equation for $\gamma_n$: 
\begin{equation}
\label{eq:pm}
p_m=2Q\left(\sqrt{\frac{3\wt{P}}{\wt{\Expt}\wt{\eps}_n^2}}\right),
\end{equation}
and hence 
\begin{equation}
\label{eq:gamma_n}
\gamma_n = \sqrt{\frac{\lambda\wt{P}-\wt{\sigma}^2}{\sigma^2_n}},
\end{equation}
where $\lambda$ is defined 
\begin{equation}
\label{eq:frack}
\lambda\dfn 3\left[Q^{-1}\left(\frac{p_m}{2}\right)\right]^{-2}
\end{equation}
\begin{remark}
$\lambda$ defined in \eqref{eq:penalty} is a special case of the above with $p_m=\frac{\Pe}{2N}$. Note again that $\lambda > \bsnr^{-1}$ must hold, which lower bounds the attainable error probability, see Remark \ref{rem:Psi2}.
\end{remark}

$\alpha$ is set so that the input power constraint at Terminal A is met. namely  $P\geq\wt{\Expt} X_n^2=\alpha^2\wt{\Expt}\wt{\varepsilon}_n^2$.
From \eqref{eq:pm} it stems that:
\begin{equation}
\alpha=\sqrt{\frac{P}{\lambda\wt{P}}}.
\end{equation}
\begin{remark}
It should be emphasized that this calculation is accurate for the coupled system only. In the original system, a modulo-aliasing error may cause the power constraint to be violated. However, since $\wt{\varepsilon}_n^2\leq 3\wt{P}$ and since the probabilities of modulo-aliasing errors are set to be very low (lower then the target error probability) the power constraint violation is negligible, and can be practically ignored; e.g., for $\Pe=10^{-6}$ and $N=20$, the increase in average power due to this effect is lower than $10^{-4} \db$. We also note in passing that the power constraint in Terminal B is always satisfied (regardless of parameter choice), due to dithering. 
\end{remark}

The parameter $\beta_n$ determines the evolution of the estimation error. 
The linear estimate of $\varepsilon_n$: $\wh{\varepsilon}_n=\beta_{n+1}Y_{n+1}$, is the optimal estimate in the coupled system, in which $\varepsilon_n$ and $Y_{n+1}$ are jointly Gaussian. We would thus like to minimize $\wt{\Expt}\left(\eps_n-\wh{\eps}_n\right)^2$. Plugging in $Y_{n+1}=\alpha(\gamma_n\varepsilon_n+\wt{Z}_n)+Z_{n+1}$ and solving the optimization
problem yields:
\begin{equation} 
\beta_{n+1}=\frac{\sigma_n}{\sigma}\frac{\sqrt{\snr\cdot\left( 1-\frac{1}{\lambda\bsnr}\right)}{ }}{1+\snr},
\end{equation}
where $\sigma_n^2$ is the variance of $\eps_n$ in the coupled system. Recalling that $\eps_{n+1} = \eps_n-\wh{\eps}_n$ and computing the MMSE for the optimal choice of $\beta_{n+1}$ above, we obtain a recursive formula for $\sigma_n^2$, which boils down to the following expression for the $\snr$ after $N$ iterations: 
\begin{equation}
\label{eq:snreqmod}
\snr_N\dfn \frac{1}{\sigma^2_N}=
\snr\cdot\left(1+\snr\cdot\frac{1-\frac{1}{\lambda\bsnr}}{1+\frac{1}{\lambda\dsnr}} \right)^{N-1},
\end{equation}
and using \eqref{eq:union_bound}, the error probability is bounded by
\begin{equation}
\label{eq:Pt}
\Pe < (N-1)p_m+2Q\left(\sqrt{\frac{3\snr_N}{2^{2NR}-1}} \right).
\end{equation}
Juxtaposing \eqref{eq:snreqmod} and \eqref{eq:snreq} shows that in the noisy feedback case, the exponential growth of the $\snr$ with the number rounds is dampened by a factor that is inversely related to $\snr$ and $\dsnr$, and also related to the term $\lambda$ that is in turn determined by the modulo-aliasing error probability. This factor corresponds to $\Psi_1,\Psi_2$ in Theorem \ref{thrm:capgap}, where $\Psi_3$ is a remainder term obtained by pedestrian manipulations and the inequality  $-\ln(1-x)\leq \frac{x}{x-1}$ for $x<1$. The result in Theorem \ref{thrm:capgap} was obtained for the specific choice $p_m = \frac{\Pe}{2N}$ of the modulo-aliasing error. In general, reducing $p_m$ decreases $\lambda$ which in turn decreases $\snr_N$, and hence increases the second addend on the right-hand-side of \eqref{eq:Pt}, resulting in a trade-off that could potentially be further optimized. 

\section{Numerical Results}
\label{results}
\begin{figure}
%
%
%
%
\begin{tikzpicture}

\begin{axis}[
xlabel={N interaction rounds},
ylabel={Capacity gap [dB]},
xmin=1, xmax=36,
ymin=0, ymax=9,
axis on top,
xmajorgrids,
ymajorgrids
]
\addplot [black, dashed]
coordinates {
(1,9.01787449938529)
(2,5.107421875)
(3,3.6474609375)
(4,2.83203125)
(5,2.314453125)
(6,1.9580078125)
(7,1.69921875)
(8,1.4990234375)
(9,1.3427734375)
(10,1.2158203125)
(11,1.11328125)
(12,1.025390625)
(13,0.947265625)
(14,0.8837890625)
(15,0.830078125)
(16,0.78125)
(17,0.732421875)
(18,0.693359375)
(19,0.6591796875)
(20,0.625)
(21,0.595703125)
(22,0.576171875)
(23,0.546875)
(24,0.52734375)
(25,0.5078125)
(26,0.48828125)
(27,0.46875)
(28,0.4541015625)
(29,0.439453125)
(30,0.4248046875)
(31,0.41015625)
(32,0.400390625)
(33,0.3857421875)
(34,0.3759765625)
(35,0.3662109375)
(36,0.3515625)

};
\addplot [black]
coordinates {
(1,9.01787449938529)
(2,7.6708984375)
(3,7.255859375)
(4,7.0263671875)
(5,6.89453125)
(6,6.81640625)
(7,6.7626953125)
(8,6.728515625)
(9,6.7041015625)
(10,6.689453125)
(11,6.6796875)
(12,6.669921875)
(13,6.6650390625)
(14,6.66015625)
(15,6.66015625)
(16,6.66015625)
(17,6.66015625)
(18,6.66015625)
(19,6.66015625)
(20,6.66015625)
(21,6.6650390625)
(22,6.669921875)
(23,6.669921875)
(24,6.6748046875)
(25,6.6748046875)
(26,6.6796875)
(27,6.6796875)
(28,6.6845703125)
(29,6.689453125)
(30,6.689453125)
(31,6.6943359375)
(32,6.69921875)
(33,6.69921875)
(34,6.7041015625)
(35,6.708984375)
(36,6.708984375)

};
\addplot [black, mark=*, mark size=3, only marks]
coordinates {
(6,6.81640625)

};
\addplot [black]
coordinates {
(1,9.01787449938529)
(2,6.494140625)
(3,5.6201171875)
(4,5.1513671875)
(5,4.86328125)
(6,4.677734375)
(7,4.541015625)
(8,4.4482421875)
(9,4.375)
(10,4.31640625)
(11,4.2724609375)
(12,4.23828125)
(13,4.208984375)
(14,4.1796875)
(15,4.16015625)
(16,4.140625)
(17,4.130859375)
(18,4.1162109375)
(19,4.1015625)
(20,4.0966796875)
(21,4.0869140625)
(22,4.08203125)
(23,4.072265625)
(24,4.0673828125)
(25,4.0625)
(26,4.0625)
(27,4.0576171875)
(28,4.052734375)
(29,4.052734375)
(30,4.0478515625)
(31,4.04296875)
(32,4.04296875)
(33,4.04296875)
(34,4.04296875)
(35,4.04296875)
(36,4.04296875)

};
\addplot [black, mark=*, mark size=3, only marks]
coordinates {
(12,4.23828125)

};
\addplot [black]
coordinates {
(1,9.01787449938529)
(2,5.29296875)
(3,3.916015625)
(4,3.1494140625)
(5,2.666015625)
(6,2.3388671875)
(7,2.099609375)
(8,1.9189453125)
(9,1.77734375)
(10,1.66015625)
(11,1.5625)
(12,1.484375)
(13,1.416015625)
(14,1.3623046875)
(15,1.30859375)
(16,1.26953125)
(17,1.23046875)
(18,1.19140625)
(19,1.162109375)
(20,1.1328125)
(21,1.11328125)
(22,1.0888671875)
(23,1.064453125)
(24,1.0498046875)
(25,1.03515625)
(26,1.015625)
(27,1.0009765625)
(28,0.986328125)
(29,0.9765625)
(30,0.9619140625)
(31,0.947265625)
(32,0.9375)
(33,0.927734375)
(34,0.91796875)
(35,0.9130859375)
(36,0.908203125)

};
\addplot [black, mark=*, mark size=3, only marks]
coordinates {
(22,1.0888671875)

};
\addplot [black]
coordinates {
(1,9.01787449938529)
(2,5.126953125)
(3,3.6767578125)
(4,2.8662109375)
(5,2.353515625)
(6,2.001953125)
(7,1.73828125)
(8,1.54296875)
(9,1.38671875)
(10,1.259765625)
(11,1.162109375)
(12,1.07421875)
(13,0.99609375)
(14,0.9375)
(15,0.87890625)
(16,0.830078125)
(17,0.7861328125)
(18,0.7470703125)
(19,0.712890625)
(20,0.68359375)
(21,0.654296875)
(22,0.625)
(23,0.60546875)
(24,0.5810546875)
(25,0.5615234375)
(26,0.546875)
(27,0.52734375)
(28,0.5078125)
(29,0.498046875)
(30,0.478515625)
(31,0.46875)
(32,0.458984375)
(33,0.4443359375)
(34,0.4296875)
(35,0.419921875)
(36,0.41015625)

};
\addplot [black, mark=*, mark size=3, only marks]
coordinates {
(23,0.60546875)

};
\path [draw=black, fill opacity=0] (axis cs:13,9)--(axis cs:13,9);

\path [draw=black, fill opacity=0] (axis cs:36,13)--(axis cs:36,13);

\path [draw=black, fill opacity=0] (axis cs:13,0)--(axis cs:13,0);

\path [draw=black, fill opacity=0] (axis cs:1,13)--(axis cs:1,13);

\node at (axis cs:10,0.2)[
  scale=0.6,
  anchor=base west,
  text=black,
  rotate=0.0
]{\itshape clean feedback};
\node at (axis cs:28,6.769453125)[
  scale=0.6,
  anchor=base west,
  text=black,
  rotate=0.0
]{ $\Delta\mathrm{SNR}=6dB$};
\node at (axis cs:6.3,6.8426953125)[
  scale=0.6,
  anchor=base west,
  text=black,
  rotate=0.0
]{ $n_{opt}=6$};
\node at (axis cs:28,4.132734375)[
  scale=0.6,
  anchor=base west,
  text=black,
  rotate=0.0
]{ $\Delta\mathrm{SNR}=10dB$};
\node at (axis cs:12.3,4.288984375)[
  scale=0.6,
  anchor=base west,
  text=black,
  rotate=0.0
]{ $n_{opt}=12$};
\node at (axis cs:28,1.0565625)[
  scale=0.6,
  anchor=base west,
  text=black,
  rotate=0.0
]{ $\Delta\mathrm{SNR}=20dB$};
\node at (axis cs:22.3,1.144453125)[
  scale=0.6,
  anchor=base west,
  text=black,
  rotate=0.0
]{ $n_{opt}=22$};
\node at (axis cs:28,0.578046875)[
  scale=0.6,
  anchor=base west,
  text=black,
  rotate=0.0
]{ $\Delta\mathrm{SNR}=30dB$};
\node at (axis cs:23.3,0.6610546875)[
  scale=0.6,
  anchor=base west,
  text=black,
  rotate=0.0
]{ $n_{opt}=23$};
\end{axis}

\end{tikzpicture}
\caption{\label{fig:resultsFigR1}
The capacity gap as function of the iterations and $\dsnr$ for a target rate $R=1$ (low $\snr$), 
and target error probability $p_t=10^{-6}$}
\end{figure}
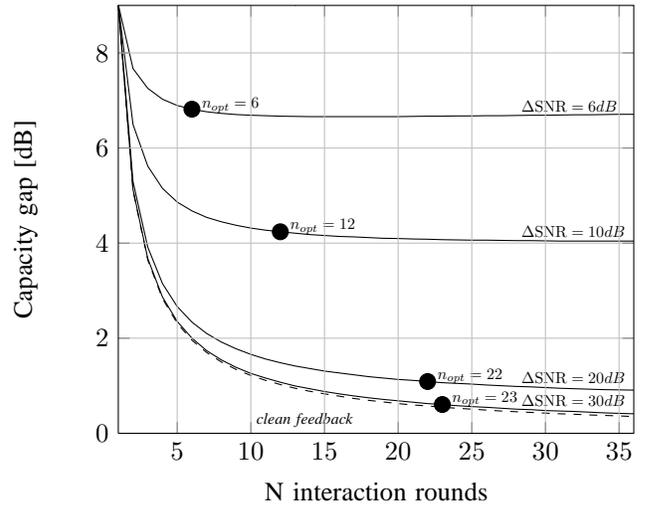

\begin{figure}
%
%
%
%
\begin{tikzpicture}

\begin{axis}[
xlabel={N interaction rounds},
ylabel={Capacity gap [dB]},
xmin=1, xmax=36,
ymin=0, ymax=9,
axis on top,
xmajorgrids,
ymajorgrids
]
\addplot [black, dashed]
coordinates {
(1,9.01787449938529)
(2,4.228515625)
(3,2.822265625)
(4,2.119140625)
(5,1.6943359375)
(6,1.416015625)
(7,1.2109375)
(8,1.0595703125)
(9,0.9423828125)
(10,0.849609375)
(11,0.771484375)
(12,0.7080078125)
(13,0.654296875)
(14,0.60546875)
(15,0.56640625)
(16,0.5322265625)
(17,0.5029296875)
(18,0.4736328125)
(19,0.44921875)
(20,0.4248046875)
(21,0.4052734375)
(22,0.390625)
(23,0.37109375)
(24,0.3564453125)
(25,0.341796875)
(26,0.3271484375)
(27,0.3173828125)
(28,0.302734375)
(29,0.29296875)
(30,0.283203125)
(31,0.2734375)
(32,0.2685546875)
(33,0.2587890625)
(34,0.25390625)
(35,0.244140625)
(36,0.2392578125)

};
\addplot [black]
coordinates {
(1,9.01787449938529)
(2,7.8173828125)
(3,7.734375)
(4,7.7197265625)
(5,7.7294921875)
(6,7.744140625)
(7,7.763671875)
(8,7.783203125)
(9,7.802734375)
(10,7.8173828125)
(11,7.83203125)
(12,7.8515625)
(13,7.861328125)
(14,7.8759765625)
(15,7.890625)
(16,7.9052734375)
(17,7.9150390625)
(18,7.9296875)
(19,7.939453125)
(20,7.94921875)
(21,7.958984375)
(22,7.96875)
(23,7.978515625)
(24,7.9833984375)
(25,7.9931640625)
(26,8.0029296875)
(27,8.0078125)
(28,8.017578125)
(29,8.02734375)
(30,8.0322265625)
(31,8.037109375)
(32,8.046875)
(33,8.0517578125)
(34,8.056640625)
(35,8.06640625)
(36,8.0712890625)

};
\addplot [black, mark=*, mark size=3, only marks]
coordinates {
(4,7.7197265625)

};
\addplot [black]
coordinates {
(1,9.01787449938529)
(2,6.69921875)
(3,6.2158203125)
(4,6.005859375)
(5,5.888671875)
(6,5.8203125)
(7,5.78125)
(8,5.751953125)
(9,5.732421875)
(10,5.72265625)
(11,5.712890625)
(12,5.7080078125)
(13,5.703125)
(14,5.703125)
(15,5.703125)
(16,5.703125)
(17,5.7080078125)
(18,5.7080078125)
(19,5.712890625)
(20,5.712890625)
(21,5.712890625)
(22,5.7177734375)
(23,5.72265625)
(24,5.72265625)
(25,5.7275390625)
(26,5.732421875)
(27,5.732421875)
(28,5.7373046875)
(29,5.7421875)
(30,5.7421875)
(31,5.7470703125)
(32,5.751953125)
(33,5.751953125)
(34,5.7568359375)
(35,5.76171875)
(36,5.76171875)

};
\addplot [black, mark=*, mark size=3, only marks]
coordinates {
(5,5.888671875)

};
\addplot [black]
coordinates {
(1,9.01787449938529)
(2,5.556640625)
(3,4.6630859375)
(4,4.2333984375)
(5,3.984375)
(6,3.828125)
(7,3.7158203125)
(8,3.6376953125)
(9,3.5791015625)
(10,3.53515625)
(11,3.49609375)
(12,3.466796875)
(13,3.4423828125)
(14,3.427734375)
(15,3.408203125)
(16,3.3935546875)
(17,3.3837890625)
(18,3.3740234375)
(19,3.3642578125)
(20,3.359375)
(21,3.3544921875)
(22,3.349609375)
(23,3.3447265625)
(24,3.33984375)
(25,3.3349609375)
(26,3.3349609375)
(27,3.330078125)
(28,3.330078125)
(29,3.330078125)
(30,3.3251953125)
(31,3.3251953125)
(32,3.3251953125)
(33,3.3203125)
(34,3.3203125)
(35,3.3203125)
(36,3.3203125)

};
\addplot [black, mark=*, mark size=3, only marks]
coordinates {
(11,3.49609375)

};
\addplot [black]
coordinates {
(1,9.01787449938529)
(2,4.404296875)
(3,3.06640625)
(4,2.40234375)
(5,2.0068359375)
(6,1.73828125)
(7,1.552734375)
(8,1.4111328125)
(9,1.3037109375)
(10,1.2158203125)
(11,1.1474609375)
(12,1.0888671875)
(13,1.0400390625)
(14,0.99609375)
(15,0.95703125)
(16,0.927734375)
(17,0.8984375)
(18,0.8740234375)
(19,0.8544921875)
(20,0.830078125)
(21,0.8154296875)
(22,0.80078125)
(23,0.78125)
(24,0.771484375)
(25,0.7568359375)
(26,0.7470703125)
(27,0.7373046875)
(28,0.7275390625)
(29,0.7177734375)
(30,0.7080078125)
(31,0.703125)
(32,0.693359375)
(33,0.6884765625)
(34,0.68359375)
(35,0.673828125)
(36,0.6689453125)

};
\addplot [black, mark=*, mark size=3, only marks]
coordinates {
(19,0.8544921875)

};
\addplot [black]
coordinates {
(1,9.01787449938529)
(2,4.248046875)
(3,2.8466796875)
(4,2.1484375)
(5,1.728515625)
(6,1.4453125)
(7,1.25)
(8,1.0986328125)
(9,0.9814453125)
(10,0.888671875)
(11,0.810546875)
(12,0.7470703125)
(13,0.693359375)
(14,0.6494140625)
(15,0.60546875)
(16,0.576171875)
(17,0.5419921875)
(18,0.517578125)
(19,0.48828125)
(20,0.46875)
(21,0.44921875)
(22,0.4296875)
(23,0.4150390625)
(24,0.400390625)
(25,0.3857421875)
(26,0.37109375)
(27,0.361328125)
(28,0.3515625)
(29,0.341796875)
(30,0.33203125)
(31,0.322265625)
(32,0.3125)
(33,0.302734375)
(34,0.2978515625)
(35,0.29296875)
(36,0.283203125)

};
\addplot [black, mark=*, mark size=3, only marks]
coordinates {
(20,0.46875)

};
\path [draw=black, fill opacity=0] (axis cs:13,9)--(axis cs:13,9);

\path [draw=black, fill opacity=0] (axis cs:36,13)--(axis cs:36,13);

\path [draw=black, fill opacity=0] (axis cs:13,0)--(axis cs:13,0);

\path [draw=black, fill opacity=0] (axis cs:1,13)--(axis cs:1,13);

\node at (axis cs:10,0.2)[
  scale=0.6,
  anchor=base west,
  text=black,
  rotate=0.0
]{\itshape clean feedback};
\node at (axis cs:28,8.10734375)[
  scale=0.6,
  anchor=base west,
  text=black,
  rotate=0.0
]{ $\Delta\mathrm{SNR}=3dB$};
\node at (axis cs:4.3,7.8094921875)[
  scale=0.6,
  anchor=base west,
  text=black,
  rotate=0.0
]{ $n_{opt}=4$};
\node at (axis cs:28,5.8221875)[
  scale=0.6,
  anchor=base west,
  text=black,
  rotate=0.0
]{ $\Delta\mathrm{SNR}=6dB$};
\node at (axis cs:5.3,5.9003125)[
  scale=0.6,
  anchor=base west,
  text=black,
  rotate=0.0
]{ $n_{opt}=5$};
\node at (axis cs:28,3.410078125)[
  scale=0.6,
  anchor=base west,
  text=black,
  rotate=0.0
]{ $\Delta\mathrm{SNR}=10dB$};
\node at (axis cs:11.3,3.546796875)[
  scale=0.6,
  anchor=base west,
  text=black,
  rotate=0.0
]{ $n_{opt}=11$};
\node at (axis cs:28,0.7977734375)[
  scale=0.6,
  anchor=base west,
  text=black,
  rotate=0.0
]{ $\Delta\mathrm{SNR}=20dB$};
\node at (axis cs:19.3,0.910078125)[
  scale=0.6,
  anchor=base west,
  text=black,
  rotate=0.0
]{ $n_{opt}=19$};
\node at (axis cs:28,0.421796875)[
  scale=0.6,
  anchor=base west,
  text=black,
  rotate=0.0
]{ $\Delta\mathrm{SNR}=30dB$};
\node at (axis cs:20.3,0.52921875)[
  scale=0.6,
  anchor=base west,
  text=black,
  rotate=0.0
]{ $n_{opt}=20$};
\end{axis}

\end{tikzpicture}
\caption{\label{fig:resultsFigR4}
The capacity gap as function of the iterations and $\dsnr$ for a target rate $R\geq4$ (high $\snr$),
and target error probability $p_t=10^{-6}$}
\end{figure}

The behavior of the capacity gap for our scheme as a function of the number of interaction rounds and $\dsnr$ is depicted in Fig.~\ref{fig:resultsFigR1} and 
Fig.~\ref{fig:resultsFigR4}, for ``high $\snr$'' and ``low $\snr$'' setups. In both figures we plotted the capacity gap, for a target rate $R$ and a target
error probability $p_e=10^{-6}$, where the $\snr$ corresponding to $R$ was found by numeric search on \eqref{eq:targetRate}, and the capacity gap calculated by definition \eqref{eq:capGapDef}. We can see that the higher $\dsnr$, the smaller the capacity gap, where $\dsnr=30\db$ is close to noiseless feedback. The points
marked $n_{opt}$ are those for which the capacity gap is less than $0.2\db$ above the minimal value attained. In Fig.~\ref{fig:resultsFigR1}, $R=1$, and can see that $\dsnr=10\db$ reduces the capacity gap to 
$4.2\db$ in 12 iterations, and $\dsnr=20\db$ reduces the capacity gap to $1.1\db$ in 22 iterations. 
In Fig.~\ref{fig:resultsFigR4} $R=4$ and for $\dsnr=10\db$ the capacity gap to $3.5\db$ in 11 iterations, and 
$\dsnr=20\db$ reduces the capacity gap to $0.8\db$ in 19 iterations. Observing \eqref{eq:capgapapprox} we can see that
for high $\snr$ the result is only a function of $\dsnr$, thus does not depend on the target rate or the
base $\snr$. 

\section{Notes on Implementation}\label{sec:implem}
The scheme described in this paper is simple and practical, as opposed to its noiseless feedback counterparts. This provides impetus for further discussing implementation related aspects. The following conditions should be met for our results to carry merit: 1) \textit{Information asymmetry}: Terminal A has substantially more information to convey than Terminal B; 2) \textit{$\snr$ asymmetry}: The $\snr$ of the feedforward channel is lower than the $\snr$ of the feedback channel. This can happen due to differences in power constraints and/or path losses; 3) \textit{Complexity/delay constraints:} There are severe complexity or delay constraint at Terminal A; 4) \textit{Two-way signaling:} Our scheme assumes sample-wise feedback. The communication system should therefore be full duplex where both terminals have virtually the same signaling rate; hence, the terminals split the bandwidth between them even though only Terminal A is transmitting information. This situation can sometimes be inherent to the system, but should otherwise be tested against the (non-interactive) solution where the entire bandwidth is allocated to Terminal A. This choice of forward vs. feedback bandwidth allocation yields a system trade-off that is $\snr$ dependent: Terminal A can use our scheme and achieve a rate of $C(\snr_{dB}-\Gamma^*_{dB})$, or alternatively employ FEC over the full forward--feedback bandwidth, thereby doubling the forward signaling rate but also incurring a $3\db$ loss in $\snr$ and a potentially larger capacity gap $\Gamma_\db^{\textrm{FEC}}$, resulting in a rate of $2C(\snr_{dB}-3\db -\Gamma_\db^{\textrm{FEC}})$. It can therefore be seen that our solution is generally better for low enough $\snr$. For instance, for $\Pe=10^{-6}$ and $\dsnr > 30\db$ our scheme outperforms (with comparable complexity and delay) full bandwidth uncocded PAM for any $\snr< 23dB$, and outperforms (with significantly smaller complexity and delay) full bandwidth FEC with $\Gamma_\db^{\textrm{FEC}} = 3\db$ for any $\snr < 9\db$.  

The use of very large PAM constellations, whose size is exponential in the product of rate and interaction rounds, seemingly requires extremely low noise and distortion at the digital and analog circuits in Terminal A, which may appear to impose a major implementation obstacle. Fortunately, this is not the case. The full resolution implied by the constellation size is by construction confined \textit{only} to the original message $\Theta$ and the final estimate $\wh{\Theta}_N$; the transmitted and received signals in the course of interaction can be safely quantized at a resolution determined only by the channel noise (and \textit{not} the final estimation noise), as in commonplace communication systems. Figuratively speaking, the source bits are \textit{revealed} along the interaction process, where the number of bits revealed in every round is determined by the channel $\snr$. This desirable property has also been confirmed in simulations.  

Another important implementation issue is sensitivity to model assumptions. We have successfully verified the robustness of the proposed scheme in several reasonable scenarios including correlative noise, excess quantization noise, and multiplicative channel estimation noise. The universality of the scheme and its performance for a wider range of models remains to be further investigated. 

\section{Discussion}\label{sec:discussion}
Note that so far we have limited our discussion to the PAM symbol error rate $\Pe$. The bit-error rate is in fact lower, since an error in PAM decoding affects only a single bit with high probability \eqref{eq:pb}, assuming Gray labeling. However, note that the modulo-aliasing error will typically result in many erroneous bits, and hence optimizing the bit error rate does not yield a major improvement over its upper bound $\Pe$. Further fine-tuning of the scheme can be obtained by non-uniform power allocation over interaction rounds in both Terminal A and B; in particular, we note that Terminal B is silent in the last round, which can be trivially leveraged. We also note in passing that our scheme can be used in conjunction with FEC as an outer code, to achieve other power/delay/complexity/error probability tradeoffs. 

We note again that for any choice of $\snr$ and $\dsnr$, the error probability attained by our scheme cannot be made to vanish with the number interaction rounds while maintaining a non-zero rate, as in the noiseless feedback S-K scheme case. 
The reason is that \eqref{eq:gamma_n} implies a minimal attainable error probability dictated by the modulo-aliasing incurred by feedback noise. Equivalently, one cannot get arbitrarily close to capacity for a given target error probability; the reason is that while increasing the number of iterations would increase $\snr_N$ and reduce the PAM decoding error term in \eqref{eq:Pt}, it would also increase the modulo-aliasing error term in \eqref{eq:Pt}. Hence, our scheme is not \textit{capacity achieving} in the usual sense. However, it can get close to capacity in the sense of reducing the capacity gap using a very short block length, typically $N\approx 20$ in the examples presented.  To the best of our knowledge, FEC schemes require a block length typically larger by two order of magnitudes to reach the same gap at the same error probability. Consequently, the encoding delay of our scheme is also markedly lower than that of competing FEC schemes. Alternatively, compared to a minimal delay uncoded system under the same error probability, our scheme operates at a much lower capacity gap for a wide regime of settings, and hence can be significantly more power efficient. 

Another important issue is that of encoding and decoding complexity. Our proposed scheme applies only a two multiplications and one modulo operation at each terminal in each interaction round. This is markedly lower than the encoding/decoding complexity of FEC, even if non-optimal methods such as iterative decoding are employed. 

\bibliographystyle{IEEEbib}
\bibliography{bibtex_references}

\end{document}